\PassOptionsToPackage{usenames,dvipsnames}{xcolor}
\documentclass[11pt,letterpaper]{article}

\usepackage[T1]{fontenc}
\usepackage[latin9]{inputenc}
\usepackage[authoryear]{natbib}
\usepackage[caption=false]{subfig}
\usepackage{ae,aecompl,geometry,float,framed,rotfloat,amsthm,soul,amsmath,
amssymb,graphicx,todonotes,xcolor,setspace,enumitem,sectsty,
booktabs,rotating,microtype,soul,units,bbm,mathtools,	palatino,bbm}

\usepackage{hyperref}
  \hypersetup{colorlinks=true,citecolor=Bittersweet,linkcolor=Bittersweet,urlcolor=Bittersweet}
\makeatletter

\newcommand*\diff{\mathop{}\!\mathrm{d}}

\theoremstyle{plain}

  \theoremstyle{plain}
  
  \theoremstyle{definition}
  \newtheorem{defn}{\protect\definitionname}
  \theoremstyle{plain}
  \newtheorem{cor}{\protect\corollaryname}
  \theoremstyle{plain}
  \newtheorem{prop}{\protect\propositionname}
  \theoremstyle{plain}
  \newtheorem*{cor*}{\protect\corollaryname}
  \theoremstyle{plain}
  
\newcounter{asscount}
\newtheoremstyle{assumption}
  {0.2cm}{0cm}
  {\rmfamily}
  {0cm}
  {\bfseries}{ }
  {0cm}
  {\thmname{#1}\thmnumber{ #2}:\thmnote{ #3}}
\theoremstyle{assumption}

\newtheoremstyle{prediction}
  {0.2cm}{0cm}
  {\rmfamily}
  {0cm}
  {\bfseries}{ }
  {0cm}
  {\thmname{#1}\thmnumber{ #2}:\thmnote{ #3}}
\theoremstyle{prediction}

\@ifundefined{showcaptionsetup}{}{%
 \PassOptionsToPackage{caption=false}{subfig}}
\usepackage{subfig}
\makeatother
  \providecommand{\corollaryname}{Corollary}
  \providecommand{\definitionname}{Definition}
  \providecommand{\lemmaname}{Lemma}
  \providecommand{\remarkname}{Remark}
  \providecommand{\propositionname}{Proposition}
\providecommand{\theoremname}{Theorem}

\newcommand{\minus}{\scalebox{0.6}{$-$}}
\newcommand{\plus}{\scalebox{0.6}{$+$}}

\newcommand{\ask}{\mathsf{ask}}
\newcommand{\bid}{\mathsf{bid}}
\newcommand{\na}{\text{news}}

\newcommand{\hft}{\text{HFT}}
\newcommand{\qs}{\text{quote sniped}}
\newcommand{\pr}{\mathsf{Pr}}
\newcommand{\E}{\mathbb{E}}
\newcommand{\cs}{\mathcal{C}}

\usepackage{titling}
\newcommand{\lm}{\lambda_{M}}
\newcommand{\lb}{\lambda_{B}}
\newcommand{\lpcs}{\lambda_\text{PC}^\star}
\newcommand{\lods}{\lambda_\text{OD}^\star}
\newcommand{\pc}{\text{PC}}
\newcommand{\od}{\text{OD}}

\newcommand{\titlepaper}{Liquid speed: \\ On-demand fast trading at distributed exchanges}

\title{\textbf {\huge \titlepaper}}

\newcommand{\mazabstract}{\noindent Exchanges acquire excess processing capacity to accommodate trading activity surges associated with zero-sum high-frequency trader (HFT) ``duels.''  The idle capacity's opportunity cost is an externality of low-latency trading. We build a model of decentralized exchanges (DEX) with flexible capacity. On DEX, HFTs acquire speed in real-time from peer-to-peer networks. The price of speed surges during activity bursts, as HFTs simultaneously race to market. Relative to centralized exchanges, HFTs acquire more speed on DEX, but for shorter timespans. Low-latency ``sprints'' speed up price discovery without harming liquidity. Overall, speed rents decrease and fewer resources are locked-in to support zero-sum HFT trades.
\bigskip{}
}

\author{\Large Michael Brolley\thanks{Michael Brolley is affiliated with Wilfrid Laurier University, Lazaridis School of Business. E-mail: \href{mbrolley@wlu.ca}{mbrolley@wlu.ca}.} \\ Wilfrid Laurier University  \and \Large Marius Zoican\thanks{Marius Zoican (corresponding author) is affiliated with University of Toronto Mississauga and Rotman School of Management. Marius can be contacted at \href{marius.zoican@rotman.utoronto.ca}{marius.zoican@rotman.utoronto.ca}. Address: Rotman School of Management, 105 St. George Street, Toronto, Ontario; Canada M5S 3E6. We have greatly benefited from discussion on this research with Carole Comerton-Forde and Nicolas Inostroza. Marius gratefully acknowledges the Connaught Fund for a New Researcher Award.
} \\ University of Toronto}

\begin{document}

\maketitle

\vspace{-10mm}
\begin{abstract}
\mazabstract



\thispagestyle{empty}

\newpage{}
\thispagestyle{empty}
\end{abstract}

\vfill{}

\vfill{}

\pagebreak{}

\vspace*{20mm}
\begin{center}
\huge \titlepaper
\par\end{center}{\Large \par}

\vspace{12mm}

\bigskip{}
\begin{abstract}
\mazabstract

\bigskip{}

\bigskip{}

\newpage{}
\setcounter{page}{1}
\end{abstract}

\newpage
\setcounter{page}{1}

\section{Introduction \label{sec:Introduction}}

Electronic markets are driven by technology. Unlike humans, computers make trading decisions with extremely low latency, approaching the speed of light. If algorithms react simultaneously to a trading opportunity, their orders arrive at the market at the same time in so-called ``micro-bursts." Indeed, using nanosecond-level message data from Nasdaq, \citet{Menkveld2018High-FrequencyMicroscope} documents that 20\% of trades cluster in sub-millisecond intervals. However, micro-bursts are associated with higher adverse selection costs for liquidity providers, consistent with high-frequencies ``races'' between algorithms reacting to the same trading signal.

\citet{Budish2015} document that the median duration of an arbitrage opportunity dropped from 97 milliseconds in 2005 to 7 milliseconds in 2015, whereas the dollar profit per arbitrage trade did not change during the same period. In markets with time priority, a trader needs to be at least as fast as its fastest competitor to capture such opportunities. However, since short-term ``duels'' between high-frequency traders (HFTs) are essentially zero-sum games, they may even harm liquidity \citep{Menkveld2017NeedLiquidity}. \citet{BiaisFoucaultMoinas2015} argue that the ``arms race'' for speed is a socially costly investment.

The low-latency arms race between HFTs also impacts the trading venues themselves. Higher demand for trading speed generates a need for better exchange infrastructure. Trading platforms have to be able to tackle surges in market activity, that is, ``micro-bursts,'' as well as normal market conditions. To avoid order delays, exchanges invest in excess processing capacity, faster computer chips, and bigger data buffers \citep{Yesalavich2010MicroburstsTraders}. Throughout 2007, at the start of the financial crisis, the New York Stock Exchange (NYSE) increased processing capacity twice: first in February, from 17 to 38 thousand messages a second, and again in August from 38 to 64 thousand messages per second. In 2011, spurred on by the emergence of high-frequency trading, the NYSE platform was able to process four times as much: 250,000 orders every second.\footnote{Sources for this paragraph: \href{https://www.wsj.com/articles/SB119249626261060148?mod=rsswn}{After Crash, NYSE Got the Message(s)}, Wall Street Journal, October 16, 2007; and \href{https://www.pcworld.com/article/238068/how_linux_mastered_wall_street.html}{How Linux Mastered Wall Street}, PC World, August 15, 2011.} More recently, Nasdaq reports in May 2019 that the Securities Information Processor (SIP), which links the U.S. trading venue data into a single feed, has a capacity of 5.6 million messages per second.\footnote{See: \href{https://www.nasdaq.com/article/time-is-relative-where-trade-speed-matters-and-where-it-doesnt-cm1156464}{Time is Relative: Where Trade Speed Matters, and Where It Doesn't}, Nasdaq, May 30, 2019.}

Exchanges have an incentive to cater to the low-latency arms race and invest in excess capacity as they can partly extract HFT rents. Trading venues have a local monopoly over their own infrastructure: for instance, you have to pay NYSE a fee for locating your computer next to the NYSE server. \citet{Budish2019WillInnovation} estimate that co-location fees amount to  \$874M-1024M in 2018, which is of the same order of magnitude as transaction fees charged by trading venues. Therefore, technology is an important source of revenue for exchanges. When trading off cost (i.e., building expensive computer capacity) against performance (message delays for traders), exchanges are biased towards the latter, since HFTs are willing to bear the cost.

Since processing power is a scarce resource, maintaining idle capacity is costly. \citet{Menkveld2018High-FrequencyMicroscope} partitions trading days into microseconds, and finds that multi-trade microseconds occur only in 10\% of the sample.  \href{https://www.marketdatapeaks.net/rates/usa/}{MarketDataPeaks} -- a website that aggregates message rates across all U.S.-based exchanges on a daily basis -- documents that the average throughput on a random day is of 3.21M messages per second, surging to 25.2M messages per second within micro-bursts. Therefore, we estimate that as much as 90\% of the high-throughput exchange infrastructure is idle for 90\% of the typical trading day.

In this paper, we study dynamic pricing of trading infrastructure and the role it can play toward reducing the social costs associated with the low-latency arms race. We build a model of high-frequency trading following \citet{Menkveld2017NeedLiquidity}, to which we introduce a real-time market for computer power (CPU) capacity. Currently, dynamic pricing is successfully implemented in other settings, for example in electricity markets, airline reservation systems, or ride-sharing platforms. Financial markets, however, are unique as the demand for excess capacity is chronically generated by a zero-sum game between fast traders, rather than by fundamental consumer demand shocks. Further, from a technological standpoint, computer processing power for trading needs to be re-priced much more frequently, ideally at sub-second intervals. 

Recent FinTech innovations allow for low-latency dynamic pricing of exchange infrastructure. In particular, Section \ref{sec:Background} describes the architecture of decentralized exchanges (DEX) currently functioning on the Ethereum blockchain. The limit order book data and the trade matching engine are distributed as computer code in a peer-to-peer network. Each participant in the network (``miner'') has a copy of the exchange and may ``rent out'' spare computer power to process incoming orders, for a fee. When a trading surge occurs, excess demand pushes up the price of computing power, and miners allocate more CPU resources to the exchange. In normal market times, miners can re-route idle computer power towards other, more productive goals.

We compare two market design choices \citep[as in, for example][]{Boulatov2013HiddenProviders}. First, we consider a \emph{centralized exchange} where HFTs invest in low-latency technology before trading starts. This setup closely corresponds to the prevailing paradigm in which traders buy co-location services from the exchange on a subscription basis: the investment is sunk at the time of trading. Therefore, at centralized exchanges the equilibrium speed investment and price of CPU power is determined ex ante, and does not depend on real-time market conditions. Second, we consider a \emph{decentralized exchange} where HFTs compete to acquire low-latency in real-time, that is, conditional on observing a profitable trading opportunity. This setup corresponds to cloud- or blockchain-based distributed exchanges as described in Section \ref{sec:Background}.

We find that decentralized exchanges with dynamic CPU pricing improve the allocation of idle resources. While a decentralized exchange does not eliminate high-frequency ``duels'' between market makers and speculators, it shortens the time during which resources are committed to support zero-sum HFT races. In centralized exchanges, HFTs rent expensive technology on a continuous basis, poised to act on short-term opportunities. Nevertheless, the rented CPU capacity is idle for a majority of the time, and excluded from alternative productive uses. By contrast, in decentralized exchanges, HFTs rent technology on-demand from peer-to-peer CPU providers. Profitable trading opportunities simultaneously generate ``micro-bursts'' in trading activity as well as surges in the price of processing power. 

Further, dynamic pricing of trading speed does not harm liquidity and may improve the speed of price discovery. As in \citet{Budish2015}, low-latency races in our model are a zero-sum game. Consequently the bid-ask spread depends only on relative speed of traders (i.e., the probability of winning the race) and not on the absolute latency. In a symmetric equilibrium, HFTs invest equal amounts in trading speed and have equal probabilities to win the race: therefore, the bid-ask spread is the same in both centralized and decentralized market. 

A decentralized market may improve price discovery. At centralized markets, high-frequency traders must commit to rent CPU power continuously, for example via co-location fees, even when there are no trading opportunities. Since decentralized markets couple CPU resources to order execution, HFTs can spend similar resources to acquire more CPU power for shorter periods of time. 
If trading at a centralized market resembles a marathon where runners need to pace themselves, trading at DEX mirrors a sprint where speed can surge higher over short intervals. Consequently, the expected time from news until a price update (either via a speculator trade or quote update from the market maker) is lower in decentralized markets. The result is, of course, conditional on the caveat that HFTs can reach the gateway of a cloud-based exchange with the same latency as for a centralized exchange (where they may colocate).

Finally, HFTs earn lower rents on decentralized than on centralized exchanges. Even if the aggregate CPU power used on decentralized exchanges is lower, high-frequency traders face higher prices during micro-bursts when speed competition intensifies, reducing their rents during speed races around news events. Therefore, decentralized exchanges transfer value from zero-sum HFT races to the suppliers of peer-to-peer exchange infrastructure.

\section{Related literature \label{sec:literature}}

Our paper contributes to an active discussion on the role of technology in financial market design. It relates to several branches of literature in market microstructure and financial innovation (FinTech).

\paragraph{High-frequency trading.} \citet{Foucault2019IsDangerous} provide a comprehensive review of the recent literature on fast trading. In recent years, an arms race emerged for ever-faster markets: \citet{Baldauf2018FastProvision} report an average exchange-to-trader latency as low as 31 microseconds for highly liquid New York Stock Exchange instruments. 

Closest to our paper, \citet{BiaisFoucaultMoinas2015} and \citet{Budish2015} explicitly model investment in low-latency infrastructure on centralized exchanges. They argue that the speed arms race generates excessive investment in low-latency technology, but does not necessarily improve market quality. In both studies, speed investment is a binary decision (i.e., a trader can be either ``fast'' or ``slow'') and is sunk before trading starts. Further, the cost of low-latency technology is fixed and corresponds, for instance, to the cost of colocation. In our paper, we model the market for trading infrastructure. We allow for real-time acquisition and pricing of low-latency technology, allowing for speed investment to be contingent on contemporaneous trading data. We argue that such a market design, made possible by peer-to-peer decentralized platforms, leads to a better allocation of limited computer resources and limits the wasteful aspect of the low-latency trading arms race.

Recent empirical evidence suggests that benefits from faster markets flattened out. For example, \citet{Ye2013} find that a drop in exchange latency on NASDAQ from microseconds to nanoseconds reduced market depth. \citet{Shkilko2019EveryCosts} find that when rain disrupts the microwave network connection between Chicago and New York, slowing down the market, liquidity actually improves. However, traders may \emph{individually} benefit from being fast. \citet{Baron2019RiskTrading} document that fast traders earn short-term profits on market orders, consistent with order sniping. \citet{Foucault2016ToxicArbitrage} find evidence that high-frequency trading generates toxic cross-market arbitrage in foreign exchange markets.

How does the exchange infrastructure impact market quality? Exchanges may use latency investment strategically: \citet{Pagnotta2018CompetingSpeed} show that trading platforms may invest in speed as a horizontal differentiation tool to relax competition. Such strategic considerations are limited in the case of decentralized exchanges, where the infrastructure is provided competitively by atomistic participants in a peer-to-peer network. \citet{Menkveld2017NeedLiquidity} argue that low-latency exchanges promote zero-sum ``duels'' between informed HFTs and can lead to lower liquidity. 

This paper proposes decentralized markets as a solution to the low-latency arms race. There are several alternative proposals available. \citet{Budish2015} argue that a discrete-time market with frequent batch auctions eliminates the advantage of being marginally faster than competitors and stimulates price competition between high-frequency traders. \citet{Kyle2017TowardExchange} propose, at the opposite end of the spectrum, a fully-continuous exchange where traders submit buy or sell trade rates over time, rather than quantities. Speed bumps, that is intentional delays to order execution, are a solution particularly favored by exchanges \citep[see][for a list of trading platforms that implemented speed bumps]{Baldauf2019High-frequencyPerformance}; however, their effectiveness depends on implementation details. \citet{Aoyagi2019SpeedBumps} argues that a random delay can incentivize investment in trading speed and worsen adverse selection. \citet{Brolley2019OrderDelays} caution that exchanges may choose the design of speed bumps for their trading platforms to maximize exchange profits, which may not coincide with the elimination of the high-frequency arms race. Existing proposals focus therefore either on removing traders' incentives to be fast, or limiting speed directly. We argue for a market-based solution, where speed is priced in real-time, and therefore more costly within trading micro-bursts. Decentralized markets do not eliminate speed races, but rather ``localize'' them only to instances when speed is required, thereby reducing the time that technology remains idle. 

Our model follows \citet{Budish2015} and \citet{Menkveld2017NeedLiquidity} in that there is no exogenous information asymmetry between high-frequency market-makers and speculators. Indeed, \citet{BrogaardHendershottRiordan2014} find that fast traders both consume and provide liquidity. Adverse selection is generated by asynchronous arrival at an exchange with time-priority rules, where trading messages are processed in the order in which they reach the exchange. This complements a large body of literature where adverse selection stems from asymmetric information \citep[for example,][]{Glosten1985, FoucaultRosu2013}.

\paragraph{Financial technology and innovation.}  The paper also relates to a growing body of research studying the impact of financial technology (FinTech) on trading and market structure.

A majority of Blockchain-driven trading protocols feature a decentralized infrastructure where transaction settlement is implemented by ``miners.'' \citet{Cong2019DecentralizedPools} argue that in proof-of-work blockchains, where only the fastest miner receives the settlement reward, there are strong incentives to form coalitions (or \emph{mining pools}) to share risk. The ensuing competition between ever larger mining pools inefficiently increases the energy consumption of proof-of-work-based blockchains. \citet{Basu2019TowardsCryptocurrencies} study the optimal design of miner fees on decentralized markets. They argue for a second-price auction system, which would reduce traders' incentives to bid strategically on transaction fees, and therefore lead to more predictable transaction costs. \citet{Biais2019TheTheorem} argue that information delays and software upgrades may lead to situations where miners ``fork'' a blockchain, leading to less reliable transaction records.  \citet{Easley2019} show that transaction fees are not necessarily Pareto-improving. They worry that, in the long run, waiting times and equilibrium fees could be high enough to discourage user participation. We complement these papers by focusing on distributed \emph{trading}: Miners in our model post orders to a P2P-maintained limit book rather than registering already matched trades.

A number of papers focus on the role of Blockchain in market design. \citet{KhapkoZoican2018} study the impact of distributed ledger on post-trade infrastructure, and find that real-time settlement can reduce liquidity as it increases inventory costs for intermediaries. \citet{Chiu2019Blockchain-BasedTrading} estimate that Blockchain-based settlement could lead to savings of 1-4 basis points per transaction in the U.S. corporate debt market. \citet{MalinovaPark2016} argue that Blockchain can improve trading transparency while limiting front-running risk, and therefore generate welfare gains.

\paragraph{Dynamic pricing.} Dynamic pricing, where prices reflect the time-varying demand and marginal costs, is a regular occurrence in other infrastructure exchanges, such as the market for electricity \citep{Joskow2012DynamicElectricity}. \citet{Cramer2016DisruptiveUber} use the example of the ride-sharing company Uber to show how a peer-to-peer platforms with dynamic pricing increase the efficiency of car usage, leading to lower costs for riders. \citet{Azevedo2016MatchingAge} discuss how advances in information technology can ameliorate the functioning of matching markets. In the same spirit, our paper argues that the benefits of a nimble, real-time, market for trading infrastructures can reduce costs associated with the low-latency arms race.

\section{Architecture of decentralized exchanges \label{sec:Background}}

As of 2019, several decentralized, exchanges (DEX) emerged in the realm of cryptosecurities, notably Binance DEX, IDEX, or Ethex. They are built on top of distributed application (dApp) platforms, such as Ethereum, and use tamper-proof and automatically-executed smart contracts to match trades \citep{Bhutoria2018DecentralizedDEX}. Following the definition of \citet{Cong2019BlockchainContracts}, smart contracts are digital contracts enforced through network consensus. On a fully decentralized exchange such as Ethex, the computer code required to match orders and to maintain the state of the order book at any given time runs on a peer-to-peer (P2P) network of computers, rather than on a centralized exchange server. 

A \emph{blockchain}, which is a decentralized distributed ledger technology where ``blocks'' of transactions are cryptographically ``chained'' together, is a representative example of such a peer-to-peer network.\footnote{We note, however, that P2P networks do not need to rely on a blockchain: BitTorrent, a file transfer software, is a peer-to-peer platform that does not use blockchain.} A distributed exchange can be implemented on a blockchain platform associated with a Turing-complete programming language, that is, a platform allowing for sophisticated state-contingent smart contracts.\footnote{A Turing-complete language can theoretically implement any algorithm or task achievable by a computer.} The Ethereum blockchain (paired with the ``Solidity'' computer language) is an example of such a platform; the Bitcoin blockchain, on the other hand, is not. All real-world examples of decentralized exchanges referenced in this Section are implemented on the Ethereum blockchain.

How does a fully decentralized exchange (DEX) work? A trader first needs to broadcast her intention to buy or to sell a given quantity for a certain price. To this end, she  accesses the DEX front-end through an application programming interface (API) or through a web-based platform. The trader voluntarily appends a \emph{gas} fee to the message to compensate P2P platform operators for running the exchange \citep[see, e.g.,][]{Easley2019}, and finally signs the order with her private key. Cryptographic signatures prevent miners from tampering with, or front-running, the order: In fact, \citet{Aune2017} propose that the orders themselves be encrypted.

Following the order broadcast, one or more peer-to-peer computers \citep[e.g., ``relayers'' in the language of][]{Warren20170x:Blockchain} observe the pending order, pick it up, and ``write it'' to the smart contract that encapsulates the matching engine. The first relayer to attach the order to the smart contract collects the gas fee. A higher gas fee is likely to attract more miners and therefore reduce the order processing delay. Finally, the distributed matching engine automatically performs one of two actions: either (partially) executes the new order against a resting order if the market is crossed, or adds the order to the limit book otherwise.

Figure \ref{fig:excharch} illustrates a fully decentralized exchange architecture, following \citet{Warren20170x:Blockchain} and \citet{AuroraLabs2017IDEXExchange}. There are three components to the latency of an order: First, there is a trader-specific network latency, that is the required time to broadcast an order. Second, relayer latency corresponds to the delay between order broadcast and the moment when a relayer injects the order into the smart contract. Finally, contract latency measures the time needed for network consensus and the smart contract code execution. The illustration closely mirrors Figure 1 in \citet[][p.~1194]{Menkveld2017NeedLiquidity} for a centralized exchange. We follow \citet{Menkveld2017NeedLiquidity} and focus on latency due to exchange infrastructure as opposed to network latency to reach the exchange front end.

\begin{figure}[H]
\caption{\label{fig:excharch} \textbf{Decentralized exchange architecture}}
\vspace{0.1in}
\begin{centering}
\includegraphics[width=\textwidth]{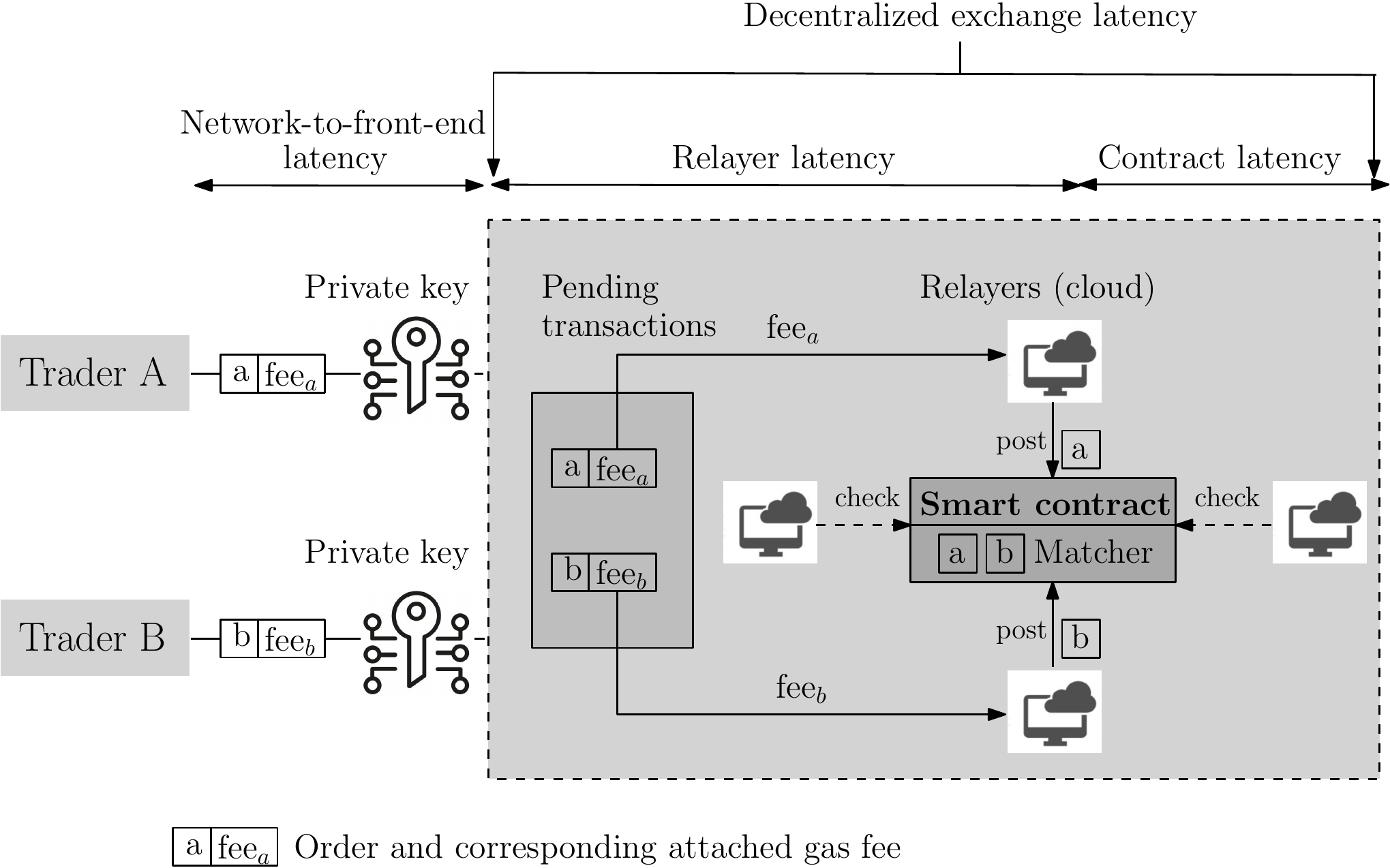}
\par\end{centering}

\end{figure}

Two design features of real-world decentralized exchanges are noteworthy. First, crytpo-asset trading platforms such as 0x, EtherDelta, or IDEX, use a hybrid ``semi-centralized'' structure featuring off-chain order relay with on-chain settlement. Limit orders are crypthographically signed and broadcast off of the blockchain with a empty counterparty field. Interested traders fill in their own crypto wallet address and inject the order into the smart contract on the blockchain. The hybrid process is faster since a trade settles on the blockchain only after counterparties are matched. Other exchanges, such as Ethex, are fully decentralized and store limit orders on the blockchain. Second, on-chain settlement is typical: in addition to the order matching smart contract, decentralized exchanges may use a second smart contract to modify traders' holdings of digital assets. In the context of cryptoassets, where exchange hacks are common, on-chain settlement augments security as it allows for direct transfers between traders' accounts \citep{AuroraLabs2017IDEXExchange}. However, a decentralized exchange implemented on a regulated, transparent platform, does not necessarily need to integrate trading and settlement.

\section{Model \label{sec:model}}

\paragraph{Asset.} A single risky asset is traded on a limit order book with price-time priority. At the start of the trading game ($t=0$), the asset has fundamental value $v$, which is common knowledge to all market participants.
Innovations to the fundamental value (``news'') arrive at a Poisson rate $\delta$. Conditional on news, the fundamental value is equally likely to be $v+\sigma$ (i.e., good news) or $v-\sigma$ (i.e., bad news), where $\sigma>0$ is the news size.

\paragraph{Trading environment.} A matching engine operates the limit order book, where it accepts limit orders and ``marketable'' limit orders submitted by market participants.  A limit order is defined as a price quote to either buy (a bid) or sell (an ask) a given amount of the asset. Any unexecuted limit orders are stored in the order book. A ``marketable'' limit order (either a limit buy order with a price higher than the lowest ask in the book or a limit sell order with a price below the highest bid) is immediately matched with the corresponding bid or ask price, resulting in a trade. We refer to such marketable limit orders simply as market orders.

\paragraph{Traders.}

There are two types of traders: $H=2$ high-frequency traders (HFTs) and a large number of liquidity investors (LI). Both trader types are risk-neutral. High-frequency traders submit both marketable and non-marketable limit orders: they compete to supply liquidity, as well as trade on fundamental value innovations. There is no fee for submitting orders to the exchange. HFTs possess a monitoring technology that allows them to perfectly and instantaneously observe changes to the security's fundamental value. Further, they can invest in technology (computer processors) allowing them to access the market faster. 

Liquidity investors receive private value shocks at a Poisson rate $\mu$. Conditional on a shock, the private value of a liquidity investor is either $\eta$ (translating to a buy intention) or $-\eta$ (sell intention), with equal probabilities, where $\eta>\sigma$. Liquidity investors are infinitely impatient and submit only market orders as soon as they are hit with a private value shock. Finally, liquidity investors do not have access to either the monitoring or market-access technologies.

\paragraph{Market Access Technology.} High frequency traders have access to a low-latency market access technology provided by suppliers of computer processing power, which we refer to as \emph{processors.}

There is a continuum of competitive processors indexed by $j$: By lending its resources to HFTs, processor $j$ incurs a linear opportunity cost of $K_j=\kappa\times j$ per unit of time. The parameter $\kappa$ measures the elasticity of computer power supply: a higher $\kappa$ translates to a higher marginal cost for technology.

The Poisson arrival rate of HFT messages is proportional to the amount of computer power she controls at a given point in time. In particular, if an HFT rents a mass $\lambda$ of processors, any submitted orders arrive at the matching engine as a Poisson process with intensity $\lambda$. 

The technology supply schedule corresponds to the marginal cost function of processors. Therefore, as Figure \ref{fig:procmarket} illustrates, if HFTs demand a processing rate $\lambda$, the clearing price for processing power is $k\lambda$, and consequently HFTs spend $k\lambda^2$ on low-latency technology.

\begin{figure}[H]
\caption{\label{fig:procmarket} \textbf{Processor market clearing}}
\vspace{0.1in}
\begin{centering}
\includegraphics[width=4in]{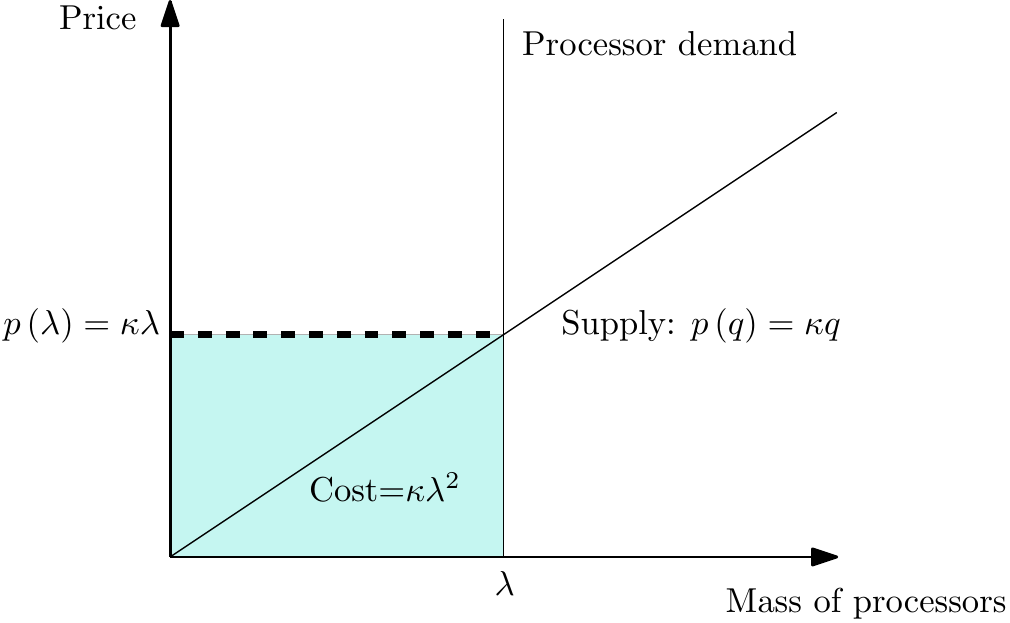}
\par\end{centering}

\end{figure}

The technology market clears at a uniform price, and quantities are allocated pro rata. In particular, if the two HFTs, $i$ and $-i$, demand order arrival rates $\lambda_i$ and $\lambda_{-i}$, the unique equilibrium technology price is $p\left(\lambda_i+\lambda_{-i}\right)=\kappa\left(\lambda_i+\lambda_{-i}\right)$. 
It follows that individual HFT investment in processing speed is 
\begin{align}\label{eq:cs}
    \cs_i&=\kappa \lambda_i\left(\lambda_i+\lambda_{-i}\right) \text{ and, respectively, }\nonumber \\
    \cs_{-i}&=\kappa \lambda_{-i}\left(\lambda_i+\lambda_{-i}\right).
\end{align}

\paragraph{Timing.} 
The timing of our setup closely follows that of \citet{Menkveld2017NeedLiquidity}. At $t=0$, the two high-frequency traders compete to provide quotes. Since any liquidity investor demands at most one unit, the HFTs submit limit orders for one unit of the asset. The highest-price buy quote (bid) and the lowest-price sell quote (ask) prevail and are posted to limit order book. We denote the HFT that succeeds in posting their quotes as the high-frequency market-maker (HFM). The other HFT acts as a quote sniper, which we refer to as the high-frequency ``bandit'' (HFB).

We define a ``trigger event'' to be the first arrival of either news or a liquidity investor. The trigger event time is random, with an expected value of $\E t_\text{trigger}=\frac{1}{\delta+\mu}$. The HFTs immediately react to the trigger event and send out orders to the exchange. With probability $\frac{\mu}{\mu+\delta}$, a liquidity investor consumes the appropriate quote that matches the direction of their trading intentions. Otherwise, with probability $\frac{\delta}{\mu+\delta}$, news arrives. The HFB rushes to submit a market order to the matching engine to snipe the stale quote, while the HFM rushes to cancel the stale quote. Upon the fill or cancellation of the stale quote, the game ends and market participants realize their payoffs. 

In the spirit of \citet{Boulatov2013HiddenProviders}, we compare two alternative exchange environments corresponding to different clearing times for the technology market. That is, we consider both
\begin{enumerate}
    \item[(a)] a centralized exchange with speed pre-commitment and,
    \item[(b)] a decentralized exchange with on-demand speed. 
\end{enumerate}
The timing of the trading game, in both cases, is illustrated below.

\begin{figure}[H]
\caption{\label{fig:timing} \textbf{Model timing}}
\vspace{0.1in}
\begin{centering}
\includegraphics[width=\textwidth]{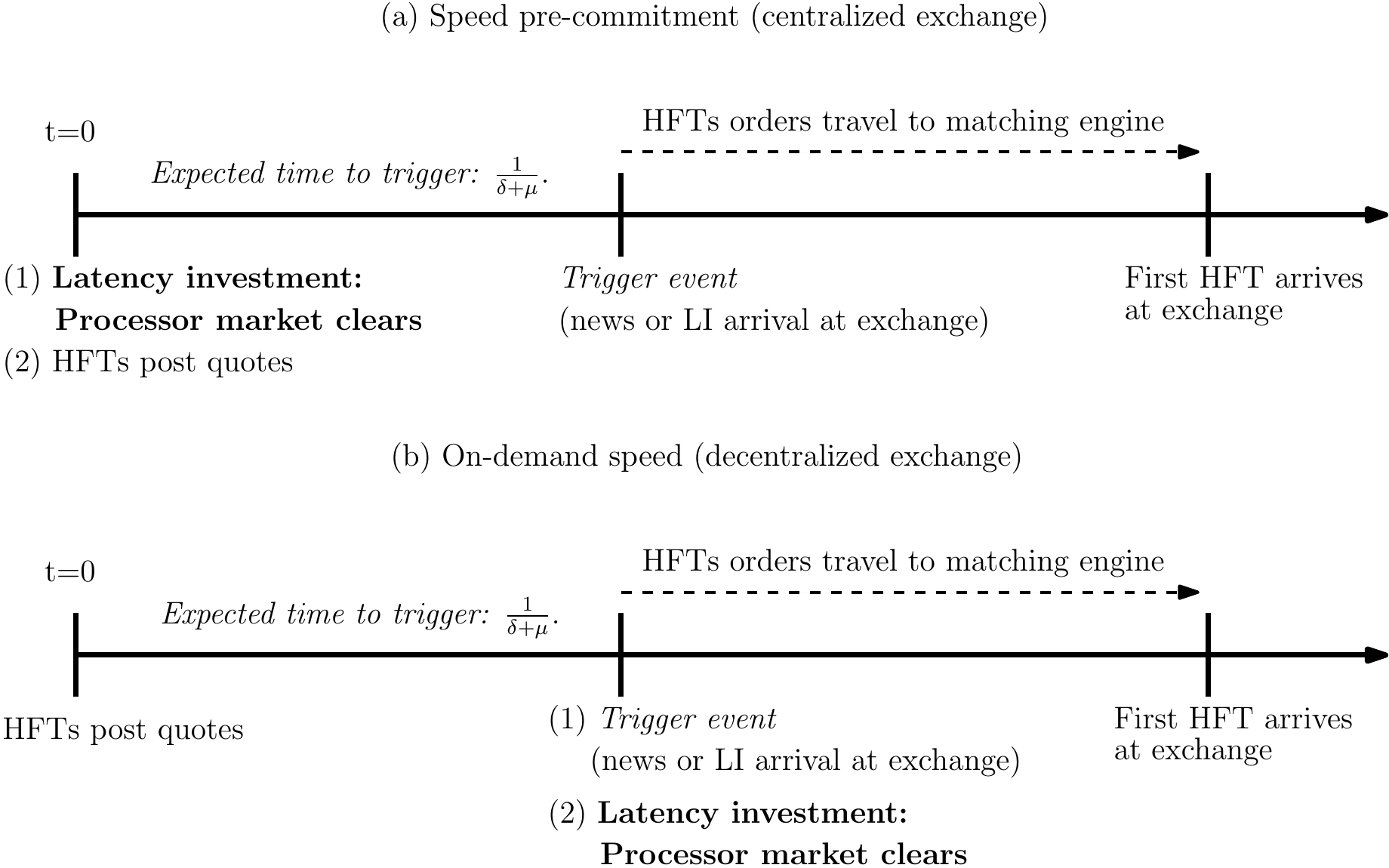}
\par\end{centering}

\end{figure}

In a \emph{speed pre-commitment} environment, the market for latency infrastructure clears at $t=0$, before trading starts. The setup echoes the prevailing market practice where centralized exchanges offer traders a menu of (usually monthly) server co-location and direct market access subscriptions. By contrast, in the \emph{on-demand speed} environment, the market for latency infrastructure clears in real-time. High-frequency traders acquire processing power ``as-needed,'' after observing the trigger event. The setup echoes the implementation Ethereum-based distributed exchanges such as EtherDelta. Traders can attach a higher gas fee to their order submissions, offering incentives for miners to prioritize their message over others. 



\section{Equilibrium\label{sec:eqm}}

We search for Nash Equilibria that are symmetric in processing speed investment ($\lambda_{i}=\lambda_{-i}$). 

\begin{defn}[Equilibrium]\label{defn1}
\begin{leftbar} \setlength{\parskip}{0ex}
An \emph{equilibrium} of the trading game consists of (i) a choice $\lambda_i$ of low-latency investment for each $\text{HFT}_i$, (ii) an HFT quoting strategy at $t=0$, that is a bid price at which each HFT is ready to buy the asset and an ask price at which each HFT is ready to sell the asset, (iii) HFT orders following the trigger event, conditional on whether the order book features his quotes at the end of $t=0$, and (iv) a market clearing price for computer processors. 
\end{leftbar}
\end{defn}

In Section \ref{sec:PC}, we study the equilibrium in a centralized exchange with speed pre-commitment; Section \ref{sec:OD} focuses on decentralized exchanges with on-demand speed acquisition.

\paragraph{Order book symmetry.} For simplicity of exposition, we focus throughout the paper on HFT ask quotes. This is without loss of generality, as buy and sell orders are symmetric as in, for example, \citet{Budish2015}.

To fix intuition, assume first good news arrived and the asset value is $v+\sigma$. A high-frequency bandit will attempt to buy the asset for $\ask$ and realize a profit of $v+\sigma-\ask$. 
Since HFTs engage in a zero-sum game, the corresponding market maker loss is $\ask-\sigma$. Alternatively, assume bad news and therefore the asset value is $v-\sigma$. In this case, the HFB attempts to sell the asset at the bid and make an identical profit of $v-\ask-\left(v+\sigma\right)=\sigma-\ask$. The two scenarios mirror each other: a successful snipe yields a profit of $\sigma-\ask$ to the HFB and a loss of $\ask-\sigma$ to the market-maker.
Similarly to this intuition, the public value of the security $v$ at $t=0$ cancels our in our analysis that follows in Sections \ref{sec:PC} and \ref{sec:OD}. 
Thus, for simplicity, we normalize $v=0$ in our subsequent analysis.
Finally, we use subscripts `M' and `B' denote variables and functions pertaining to HFM and HFB, respectively.

\subsection{Centralized exchange environment (pre-committed speed investment)\label{sec:PC}}

In this section, we assume that investment in processing speed by $\hft_{i}$ given by $\cs_{i}(\lambda)$ in equation \eqref{eq:cs} is made at the beginning of $t=0$, just prior to the start of the trading game where HFTs rush to post quotes.  
We view this pre-committment to processing speed--denoted \emph{PC} --as similar to the current centralized exchange paradigm: firms pay co-location subscription fees to access the exchange with reduced latency, where access persists for a length of time such that speed investment is a sunk cost with respect to any one trading day that begins at $t=0$.


We solve this game via backward induction. 
At the quoting race stage following speed investment at $t=0$, HFTs race to post their quotes to the matching engine to become the HFM.
Because HFTs $i$ and $-i$ are both rushing to post quotes but do not yet know their roles in the game, the expected payoff from participating as the HFM is written as,
\begin{align}
    \pi_{0}^\text{M} = & \text{ } \pr(\na)\times\pr(\qs)\times(\ask-\sigma)+(1-\pr(\na))\times\ask \nonumber
    \\ & -\E[\text{processor rental time} \mid \lambda_{i}]\times\cs_{i}(\lambda_{i}), \nonumber
    \\ = & \text{ } -\frac{\delta}{\delta+\mu} \frac{\lambda_{-i}}{\lambda_i+\lambda_{-i}} \left(\sigma-\ask\right)+\frac{\mu}{\delta+\mu} \ask-\E[\text{processor rental time} \mid \lambda_{i}]\times\cs_{i}(\lambda_{i}). \label{hfmt1pi}
\end{align}
where the term $\frac{\lambda_{-i}}{\lambda_i+\lambda_{-i}}$ reflects the probability that should $\hft_{i}$ become the HFM, $\hft_{i}$ loses the race to cancel the stale quote $\ask$ following the arrival of news, and $\E[\text{processor rental time} \mid \lambda_{i}]\times\cs_{i}(\lambda_{i})$ reflects the total expected processing expenditure (i.e., duration multiplied by speed intensity). With probability $\frac{\mu}{\delta+\mu}$, a liquidity investor arrives to the market before news, buys the asset, and the HFM receives the ask price.

If $\hft_{i}$ becomes the HFB, they will seek to snipe stale quotes in the event of news arrival. 
As such, their expected payoff is given by,
\begin{align}\label{hfbt1pi}
    \pi_{0}^\text{B} & = \pr(\na)\times\pr(\qs)\times(\sigma-\ask)-\E[\text{processor rental time} \mid \lambda_{i}]\times\cs_{i}(\lambda_{i}),\nonumber
    \\ & = \frac{\delta}{\delta+\mu} \frac{\lambda_{i}}{\lambda_{i}+\lambda_{-i}} \left(\sigma-\ask\right)-\E[\text{processor rental time} \mid \lambda_{i}]\times\cs_{i}(\lambda_{i}).
\end{align}
In equation \eqref{hfbt1pi}, since $\hft_{i}$ is the bandit, the snipe probability is $\frac{\lambda_{i}}{\lambda_{i}+\lambda_{-i}}$, that is, the probability that $\hft_{i}$ is the first to arrive at the market after news.

What is the equilibrium spread? To ensure both HFTs are indifferent to becoming an HFM or an HFB, HFTs will quote a price $\ask^\star$ such that the expected profit functions in \eqref{hfmt1pi} and \eqref{hfbt1pi} are equal. 
Intuitively, if $\hft_i$ instead posts $\ask^{\plus}>\ask^\star$, then $\hft_{-i}$ always becomes the market-maker since she quotes the tighter spread. Therefore, $\hft_i$ earns the HFB expected profit for a spread of $\ask^\star$ and is consequently indifferent between deviating or not.\footnote{The intuition here follows the proof of Proposition 1 from \citet{Menkveld2017NeedLiquidity}, p.~1217.} Similarly, if $\hft_i$ quotes a tighter spread, that is $\ask^{\minus}<\ask^\star$, she becomes the market maker for sure. However, the deviation is not profitable since
\begin{equation}
    \pi_{0}^\text{M}\left(\ask^{\minus}\right)<\alpha\pi_{0}^\text{M}\left(\ask^\star\right)+\left(1-\alpha\right)\pi_{0}^\text{B}\left(\ask^\star\right), \forall \alpha \in \left[0,1\right].
\end{equation}.

Setting (\ref{hfmt1pi}) and (\ref{hfbt1pi}) equal and solving for $\ask^\star$, we obtain,
\begin{align}
\nonumber -\frac{\delta}{\delta+\mu} \frac{\lambda_{-i}}{\lambda_i+\lambda_{-i}} \left(\sigma-\ask\right)+\frac{\mu}{\delta+\mu} \ask & = \frac{\delta}{\delta+\mu} \frac{\lambda_{i}}{\lambda_{i}+\lambda_{-i}} \left(\sigma-\ask\right)
\\ \nonumber \iff \frac{\mu}{\delta+\mu} \ask & = \frac{\delta}{\delta+\mu} \frac{\lambda_{i}+\lambda_{-i}}{\lambda_i+\lambda_{-i}} \left(\sigma-\ask\right)
\\ \iff \ask^\star & =\frac{\delta}{\delta+\mu}\sigma.\label{pc_ask}
\end{align}
Note that the cost of investment in speed from (\ref{hfmt1pi})-(\ref{hfbt1pi}) cancels out when we compare the two payoff functions at the quoting stage, as the investment is sunk once the quoting game begins at $t=0$.
In fact, equation (\ref{pc_ask}) is not only independent of the expected processing speed cost, but---as in \citet{Budish2015}---the quoted spread $\ask^\star$ is independent of speed levels altogether. The HFT ``duel'' is essentially a zero-sum game. Any investment in speed that reduces the probability of adverse selection to $\hft_{i}$ in their role as a market-maker also increases their ability to snipe a stale quote as a bandit at the same rate, $\frac{\lambda_{-i}}{(\lambda_i+\lambda_{-i})^2}$. 
Moreover, the quoting decision of $\hft_{i}$ is impacted similarly by the speed investment of $\hft_{-i}$. If $\hft_{-i}$ invests more in speed, then $\hft_{i}$'s payoff as either market maker or bandit decreases at the same rate, that is $\frac{\lambda_{i}}{(\lambda_i+\lambda_{-i})^2}$.    
Hence, the equilibrium quoted price does not depend on the initial investment in speed by both HFTs.

Prior to the quoting game, HFTs commit resources to rent market access technology from processors at the beginning of $t=0$.
Market access technology is rented until the game ends following a trigger event.
We denote the (random) trigger event time as $\tau_\text{trigger}=\min\left\{\tau_\na,\tau_\text{LI}\right\}$. 
If the trigger event is an LI arrival, then the game stops immediately following the trigger event. 
Conversely, if the trigger event is news arrival, the game continues until the event time $\tau_\text{HFT race}$ that the stale quote is either consumed or updated. 
Formally, we write the expected processor rental duration for an HFT as:
\begin{align}
    \E\left[\tau_\text{trigger}+\mathbbm{1}_{\tau_\text{news}\leq \tau_\text{LI}}\tau_\text{HFT race}  \mid \lambda_{i} \right] 
    &=\int_0^\infty \left[\int_0^y \left(x+\frac{1}{\lambda_{i}+\lambda_{-i}}\right) \delta e^{-\delta x} \diff x + y \int_y^\infty \delta e^{-\delta x} \diff x \right] \mu e^{-\mu y} \diff y, \nonumber \\ &= \frac{1}{\delta+\mu}+\frac{\delta}{\delta+\mu} \frac{1}{\lambda_{i}+\lambda_{-i}}.\label{procdur}
\end{align}
The expected duration of processor use simplifies to equation \eqref{procdur}, which sums the expected duration of two events: the expected time until the first event $\left(\frac{1}{\delta+\mu}\right)$, and the expected duration of the HFT race following the arrival of news $\left(\frac{1}{\lambda_{i}+\lambda_{-i}}\right)$, which occurs with probability $\frac{\delta}{\delta+\mu}$.

Given the equilibrium quote $\ask^\star$ from \eqref{pc_ask} and the expected processor rental duration from \eqref{procdur}, we write the $\hft_{i}$ payoff function at $t=0$ denoted $\pi_{0}^\text{i}$ as,
\begin{align}\label{hfteqpc}
    \pi_{0}^\text{i}(\lambda_{i}) & =\pr(i=M \mid \lambda_{i})\pi_{0}^{\text{M}}(\lambda_{i})+\pr(i=B \mid \lambda_{i})\pi_{0}^{\text{B}}(\lambda_{i})-\E[\text{processor rental time} \mid \lambda_{i}]\times\cs_{i}(\lambda_{i}) \nonumber
    \\ &=\frac{\lambda_i}{\lambda_{-i}+\lambda_{i}} \frac{\delta\mu\sigma}{(\delta +\mu )^2}-\left(\frac{1}{\delta+\mu}+\frac{\delta}{\delta+\mu}\frac{1}{\lambda_{i}+\lambda_{-i}}\right)\times\lambda_i \kappa\left(\lambda_i+\lambda_{-i}\right).
\end{align}
The simplification in (\ref{hfteqpc}) follows from the fact that $\pi_{0}^{\text{M}}(\ask^\star)=\pi_{0}^{\text{B}}(\ask^\star)$.
To solve for the optimal processing speed intensity $\lambda_{i}$, we take the first-order condition of (\ref{hfteqpc}) and solve for the fixed point under the assumption of symmetric investment intensities ($\lm=\lb$). Proposition \ref{prop:pc_eq} describes the equilibrium that obtains.
\begin{prop}[Symmetric Pre-Commitment Equilibrium]\label{prop:pc_eq}
Let $(\delta,\mu,\sigma,\kappa)\in(0,\infty)^4$.
There exists a unique Nash Equilibrium in the sense of Definition \ref{defn1} where $\ask^\star$ is as in Equation (\ref{pc_ask}), and  $\lambda_{i}=\lambda_{-i}=\lpcs\in(0,\infty)$ is the maximum of Equation (\ref{hfteqpc}).
Moreover, $\lpcs$ is given by:
\begin{equation}\label{leqpc}
    \lpcs=\frac{-\delta+\sqrt{\delta^2+\frac{3\delta\mu\sigma}{\kappa(\delta+\mu)}}}{6}.
\end{equation}
\end{prop}
\begin{proof}
The beginning of the proof follows from the discussion above. 
What remains is to take the first-order condition of equation (\ref{hfteqpc}) with respect to $\lambda_{i}$ and invoking symmetry of speed intensity ($\lambda_{i}=\lambda_{-i}=\lpcs$)  to show that the resulting solution $\lpcs$ is unique and positive.
\begin{align}
    \text{F.O.C (\ref{hfteqpc}):}  & \text{ } \frac{\lambda_{i}\left(\delta\mu\sigma-(\lambda_{i}+\lambda_{-i}+\delta)\kappa(\lambda_{i}+\lambda_{-i})(\delta+\mu)\right)}{(\delta+\mu)^2(\lambda_{i}-\lambda_{-i})}=0, \label{focpc}
    \\ \iff & \lpcs = \frac{-\delta+\sqrt{\delta^2+\frac{3\delta\mu\sigma}{\kappa(\delta+\mu)}}}{6}.
\end{align}
where $\lpcs$ is the unique positive root, as the second root is negative and thus inadmissible. 
\end{proof}

\subsection{Decentralized exchange environment (on-demand speed investment) \label{sec:OD}}

Consider an environment in which HFTs invest in ``on-demand'' (\emph{OD}) processing speed $\cs_{i}(\lambda)$: HFTs submit a fee to rent processing speed at the same time that they submit an order to the exchange.
This environment captures the essence of a decentralized exchange, where a group of trade facilitators---``processors'' in the language of our model---provide order-processing services for a fee.
In a decentralize exchange, an HFT submits an order to a processor with the requisite fee such that the order is ferried to the matching engine. 
In our setup, $\hft_{i}$ will choose a Poisson arrival rate $\lambda_{i}$, where the intensity of processing speed reflects the mass of processors that the HFT pays to execute their order. 
The intuition is that the greater mass of processors that the HFT rents to submit their order, the greater the probability that their order is given priority over any competitors. 

With on-demand speed investment, an investment decision is made following the quoting game, which happens only following the arrival of news. 
Similar to the centralized exchange case, the news arrival triggers a rush by the HFB to snipe the stale quote.
To improve his chances of doing so, the HFB chooses the mass of processors $\lambda_B$ to maximize the expected payoff from attempting to snipe the stale quote,
\begin{align}
    \pi_{0}^\text{B}(\lb \mid \text{news arrival}) & = \pr(\qs)\times(\sigma-\ask)-\E[\text{processor rental time} \mid \lb]\cs_{B}(\lb),\nonumber
    \\ & = \frac{\lambda_B}{\lambda_B+\lambda_M} \left(\sigma-\ask\right) - \frac{1}{\lb+\lm}\lb \kappa \left(\lb+\lm\right). \label{hfbod}
\end{align}

Similarly, the arrival of news will trigger a reaction by the HFM to rush to cancel their stale quote.
The HFM rents a mass of processors $\lambda_M$ to minimize the cost of being sniped:
\begin{align}
    \pi_{0}^\text{M}(\lm \mid \text{news arrival}) & = \pr(\qs)\times(\ask-\sigma)-\E[\text{processor rental time} \mid \lm]\cs_{M}(\lm),\nonumber
    \\ & = \frac{\lambda_B}{\lambda_B+\lambda_M} \left(\ask-\sigma\right) - \frac{1}{\lb+\lm}\lambda_M \kappa \left(\lambda_B+\lambda_M\right). \label{hfmod}
\end{align}
Taking the first-order conditions of (\ref{hfbod}) and (\ref{hfmod}), we obtain,
\begin{align}
    \text{F.O.C (\ref{hfbod}):}  & \text{ }  \frac{\lambda_M}{(\lambda_B+\lambda_M)^2} \left(\sigma-\ask\right) - \kappa = 0, \label{focbod}
    \\  \text{F.O.C (\ref{hfmod}):}  & \text{ }  \frac{\lambda_B}{(\lambda_B+\lambda_M)^2} \left(\sigma-\ask\right) - \kappa = 0.\label{focmod}
\end{align}
We note that first-order conditions (\ref{focbod}) and (\ref{focmod}) are symmetric and solve for the fixed point, $\lm^\star=\lb^\star=\lods$,
\begin{equation}\label{eq_lods}
  \lambda^\star_\text{OD}=\frac{\sigma-\ask}{4\kappa}.
\end{equation}

Next, by backward induction, we analyze the quoting decision at $t=0$, and solve for $\ask^\star$, taking into account the optimal speed investment at the trigger time $\tau_\text{trigger}$. 
An $\hft_{i}$ selects their quoting strategy conditional on the anticipated outcome of the sniping/cancelling game played by the HFM and HFB.
Similarly to Section \ref{sec:PC}, the $\hft_{i}$ chooses the quote $\ask^\star$ such that the expected payoff to becoming either an HFM or an HFB are equal.
Evaluating the payoffs $\pi_{0}^\text{M}(\lambda_{i})$ and $\pi_{0}^\text{B}(\lambda_{i})$ at $\lods$ yields,
\begin{align}
    \pi_{0}^\text{M}(\lods) &=\frac{\delta}{\delta+\mu} \frac{\left(\ask-\sigma\right)}{2} +\frac{\mu}{\delta+\mu} \ask - \frac{\delta}{\delta+\mu}\frac{\sigma-\ask}{4\kappa}\kappa.
   \\ \pi_{0}^\text{B}(\lods) &= \frac{\delta}{\delta+\mu}\frac{\left(\sigma-\ask\right)}{2} - \frac{\delta}{\delta+\mu} \frac{\sigma-\ask}{4\kappa}\kappa.
\end{align}

Solving for $\ask^\star$ such that $\pi_{0}^\text{M}(\ask^\star; \lods)=\pi_{0}^\text{B}(\ask^\star; \lods)$ we obtain,
\begin{equation}\label{od_ask}
    \ask^\star=\frac{\delta}{\delta+\mu}\sigma.
\end{equation} 
Taken together, $\ask^\ast$ into $\lambda_{OD}^{\ast}$ yield the following proposition.
\begin{prop}[Symmetric On-Demand Equilibrium]\label{prop:od_eq}
Let $(\delta,\mu,\sigma,\kappa)\in(0,\infty)^4$.
There exists a unique Nash Equilibrium in the sense of Definition \ref{defn1} where $\ask^\star$ is as in Equation (\ref{od_ask}), and  $\lambda_{i}=\lambda_{-i}=\lods\in(0,\infty)$ solves the system of first-order conditions (\ref{hfbod})-(\ref{hfmod}).
Moreover, $\lods$ is given by:
\begin{equation}\label{leqod}
     \lods=\frac{\mu\sigma}{4\kappa(\delta+\mu)}.
\end{equation}
\end{prop}
\begin{proof}
The proof follows the discussion above; equation (\ref{leqod}) results from inputting $\ask^\star$ from equation (\ref{od_ask}) into the expression for $\lods$ in equation (\ref{eq_lods}).
\end{proof}

Figure \ref{fig:speed_intensity} illustrates the result in Propositions \ref{prop:pc_eq} and \ref{prop:od_eq}. First, HFTs acquire more processing power in decentralized relative to centralized exchanges. Second, HFT invest more (less) in speed on centralized (decentralized) markets as the news arrival rate increases. A higher news rate generates two effects. On the one hand, it increases the likelihood of a sniping opportunity, and therefore the value of trading speed. On the other hand, it translates to a wider bid-ask spread and lower sniping profits, conditional on news. On a centralized exchange, the first effect dominates and the value of speed increases in news frequency. In contrast, on a decentralized exchange, speed investment is incurred \emph{conditional} on news arrival, and consequently only the second effect persists.

\begin{figure}[H]
\caption{\label{fig:speed_intensity} \textbf{HFT Speed Investment}}
\begin{minipage}[t]{1\columnwidth}%
\footnotesize
This figure illustrates the equilibrium HFT investment in low-latency technology measured by the Poisson arrival rate $\lambda$, as a function of the news arrival rate $\delta$. Parameter values: $\mu=2$, $\kappa=0.25$, and $\sigma=1$.
\end{minipage}

\vspace{0.075in}

\begin{centering}
\includegraphics[width=\textwidth]{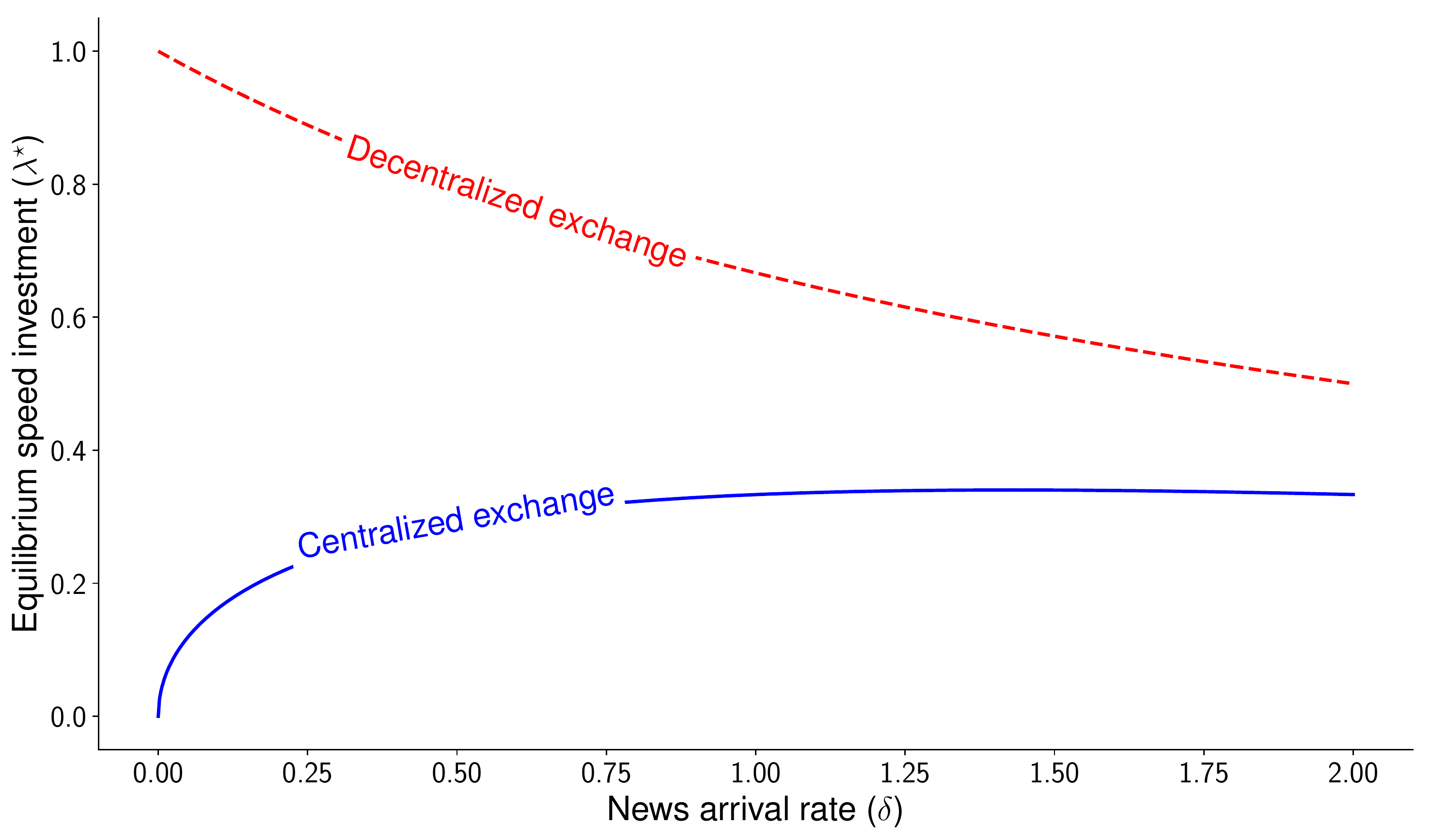}
\par\end{centering}

\end{figure}

\section{Impact on market quality \label{sec:Disc}}

In this section, we compare the equilibrium outcomes of the centralized and decentralized environments to examine how changing the timing of HFT investment in low-latency technology translates to changes in liquidity, price discovery, and the total amount of resources allocated to the speed race.

\paragraph{Liquidity and Price discovery.}
In Sections \ref{sec:PC} and \ref{sec:OD}, we obtain that both the centralized and decentralized markets lead HFTs to an identical equilibrium quoting strategy ($\ask^\star_{\pc}=\ask^\star_{\od}$) that is independent of their investment in speed intensity $\lambda$. 

\begin{cor}[Liquidity]\label{liq:pcod}
The bid-ask spread $\ask-\bid$ is equal to $2\times\frac{\delta\sigma}{\delta+\mu}$ in both the centralized and decentralized market setting.
\end{cor}

The result in Corollary \ref{liq:pcod} echoes the model of \citet{Budish2015}, who compare a continuous limit order market to frequent batch auctions. In both models, since the speed race is a zero-sum game between high-frequency traders, the bid-ask spread does not depend on absolute latency levels. Instead, the bid-ask spread depends solely on the magnitude of adverse selection costs.

Both the HFB and HFM spend resources to compete in the sniping/cancelling race that commences following the arrival of news.
In a centralized market, resource commitment to processing speed takes place before quotes are posted, leaving the possibility that HFTs will commit resources to a speed race that never begins (i.e., if a liquidity trader arrives before news). 
HFTs in a decentralized market, however, invest in low-latency technology only after a speed race is triggered by news arrival. 
Using the equilibrium values for speed intensity across the two markets, we evaluate the difference in speed intensity $\lpcs-\lods$ from equations (\ref{leqpc}) and (\ref{leqod}), as well as the difference in speed prices $p(\lpcs)-p(\lods)$ to arrive at the following proposition.

\begin{prop}
[Speed Intensity]\label{si:pcod}
The equilibrium speed intensity $\lambda^\star$ and price of speed $p(\lambda^{\star})$ are higher in the decentralized market than in the centralized market.
\end{prop}
\begin{proof}
See Appendix.
\end{proof}

Proposition \ref{si:pcod} aligns with intuition that HFTs will commit more resources to speed intensity---and pay a higher average per-unit price to do so---if they know that the speed race will occur. 
In a centralized market where HFTs must pre-commit to processing capacity knowing that said resources may remain idle, the competition for speed intensity is lower. 

Irrespective of the market environment, the outcome of the speed race is identical: stale quotes are removed from the market to make way for quotes that incorporate new information, thereby contributing to the price discovery process.
We can thus define a measure of price discovery that describes the expected race-time, and thus measure how quickly, on average, news is impounded into prices.
From the time of news arrival, the expected time until the \textit{first} order arrives is given by $\frac{1}{2\lambda^\star}$, which is the inverse of the sum of the equilibrium speed intensity investment by both HFTs.
Thus, taking the difference $\frac{1}{2\lpcs}-\frac{1}{2\lods}$,  we can infer from Proposition \ref{si:pcod} that the greater competition for speed intensity ($\lods>\lpcs$) in the decentralized market leads to faster price discovery relative to the centralized market.
We summarize our price discovery result in Corollary \ref{pd:pcod} below.

\begin{cor}
[Price Discovery]\label{pd:pcod}
The expected time required for news to impound into prices is shorter in a decentralized market than in a centralized market.
\end{cor}

Figure \ref{fig:price_discovery} illustrates the result in Corollary \ref{pd:pcod}: price discovery is faster on decentralized exchanges since HFTs engage in a more intense arms race (albeit for shorter intervals). One caveat remains that since decentralized exchange do not allow for co-location (since the infrastructure is cloud-based), the trader-to-exchange latency is higher than at centralized exchanges. Corollary \ref{pd:pcod} implies that this effect is at least partly compensated by higher investments in on-demand speed on decentralized exchanges.

\begin{figure}[t]
\caption{\label{fig:price_discovery} \textbf{Speed of Price Discovery}}
\begin{minipage}[t]{1\columnwidth}%
\footnotesize
This figure illustrates the expected time elapsed between a news event and the next trade or quote update, in both centralized and decentralized exchanges, as a function of the news arrival rate $\delta$. Parameter values: $\mu=2$, $\kappa=0.25$, and $\sigma=1$.
\end{minipage}

\vspace{0.075in}

\begin{centering}
\includegraphics[width=\textwidth]{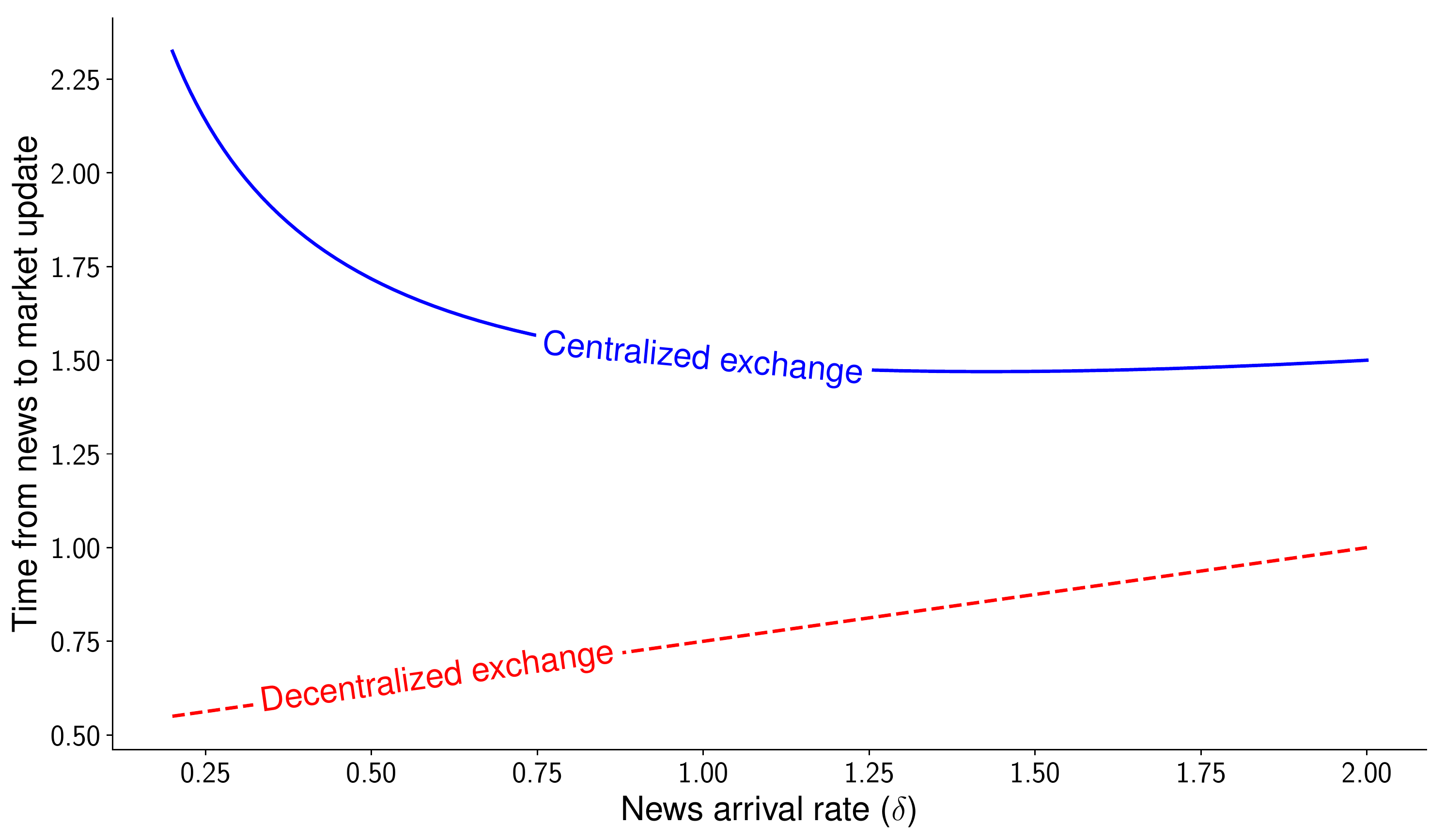}
\par\end{centering}

\end{figure}

\paragraph{Total Resource Usage and Rents from the Speed Race.}

Though a decentralized market leads the competition for processing speed to intensify at the point when processing resources are acquired, the length of time for which these resources need to be retained is a key advantage of a decentralized market: an HFT rents resources to process a single order, and upon order completion or cancellation, the resources are released. 

Using the equilibrium HFT investment in processing speed intensity $\lpcs$, we compute total resource usage in a centralized exchange environment. 
We denote total resource usage as $\Lambda_{PC}^\star$, which computes twice the per-HFT processor speed intensity multiplied by the average processor rental time $\E [\text{processor rental time} \mid \lpcs]$:
\begin{align}\label{eqtimepc}
    \Lambda_{PC}^\star & \equiv 2\lpcs \times  \E [\text{processor rental time} \mid \lpcs], \nonumber
    \\ & = 2\lpcs \times \left(\frac{1}{\delta+\mu}+\frac{\delta}{\delta+\mu}\frac{1}{2\lpcs}\right), \nonumber
    \\ & = \frac{\delta}{\delta+\mu}+\frac{\sqrt{\delta^2+\frac{3\delta\mu\sigma}{\kappa(\delta+\mu)}}-\delta}{3(\delta+\mu)}.
\end{align}
On-demand total resource usage, denoted by $\Lambda_\text{OD}^\ast$, is given by,
\begin{align}\label{eqtimeod}
    \Lambda_{\od}^\ast & = 2\lods \times \left(\frac{\delta}{\delta+\mu}\frac{1}{2\lods}\right) = \frac{\delta}{\delta +\mu}.
\end{align}

Because speed intensity and time-to-execution following news arrival are perfectly inversely-related, resource usage in the event of news is identical in both environments, regardless of the investment in speed intensity. On a centralized exchange, however, pre-commitment generates excessive resource consumption: both HFTs rent a mass $\lpcs$ of processors which are rented-but-idle from $t=0$ until the a trigger event occurs, a period with an expected length of $\frac{1}{\delta+\mu}$.

\begin{cor}[Resource Usage]\label{re:pcod}
The expected total resource allocation to processing speed is lower in the decentralized market than the centralized market, $\Lambda_{\pc}^{\star}>\Lambda_{\od}^{\star}$.
\end{cor}

\begin{figure}[H]
\caption{\label{fig:resource} \textbf{Total Infrastructure Committed to Low-Latency Trading}}
\begin{minipage}[t]{1\columnwidth}%
\footnotesize
This figure illustrates the total usage of exchange infrastructure resources, in both centralized and decentralized markets, as a function of the news arrival rate $\delta$. Parameter values: $\mu=2$, $\kappa=0.25$, and $\sigma=1$.
\end{minipage}

\vspace{0.075in}

\begin{centering}
\includegraphics[width=\textwidth]{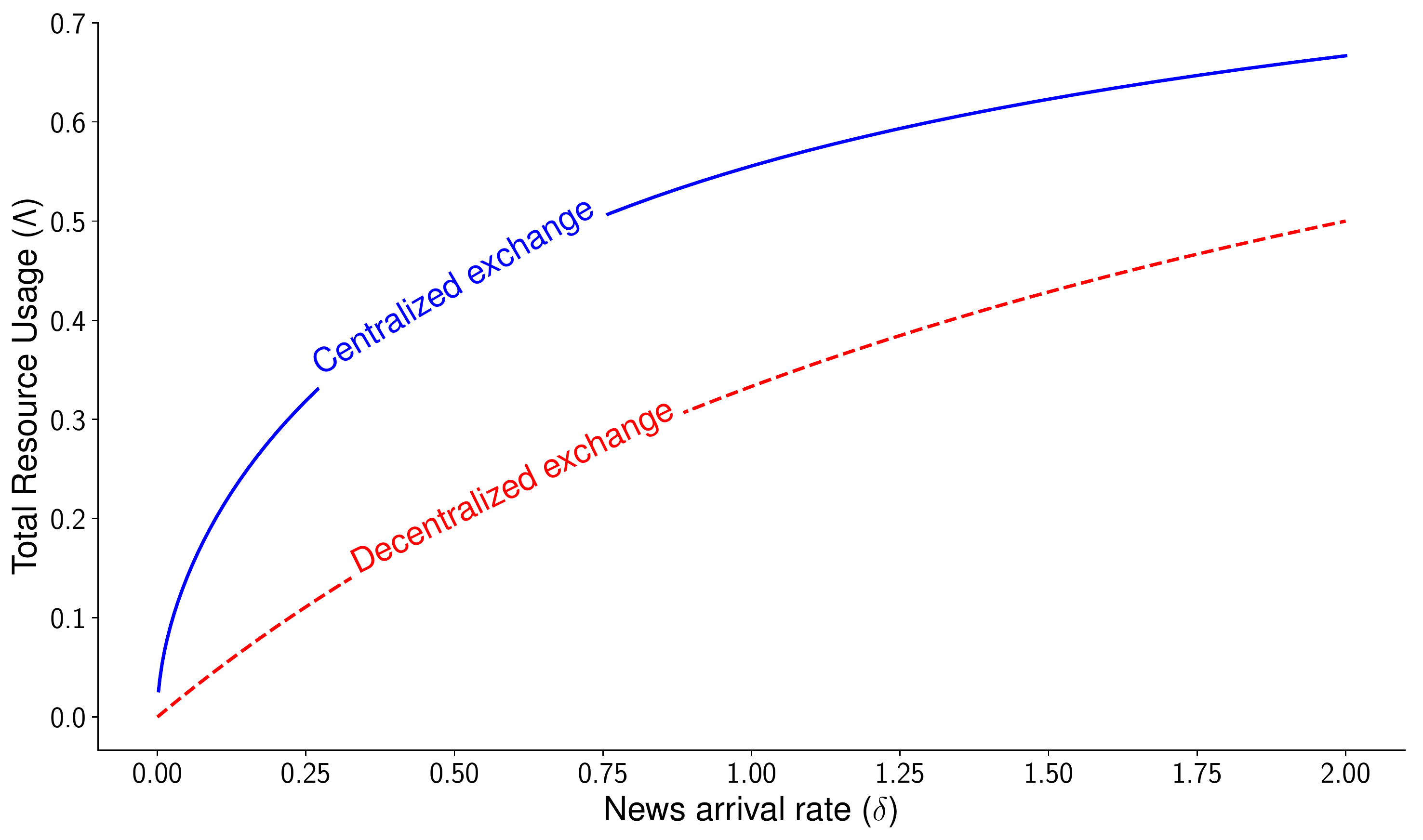}
\par\end{centering}

\end{figure}

Figure \ref{fig:resource} illustrates Corollary \ref{re:pcod}. As the news arrival rate increases, the expected resource usage increases in both centralized and decentralized markets, as HFT duels become more likely. 
On the one hand, HFTs acquire more CPU power upon observing news in decentralized markets; however, infrastructure is used more efficiently, with CPU power being rented only when a sniping opportunity emerges. 
We predict that the second effect dominates, yielding lower resource usage at decentralized exchanges.

Finally, we evaluate how on-demand speed in decentralized markets impacts the rents earned by high frequency traders from the speed race.
First, we evaluate the equilibrium HFT profit $\pi^{i}_{0}$ within the centralized market environment, from which we obtain:
\begin{align}\label{hftpipc1}
    \pi_{0}^\text{i}(\lpcs) & = \frac{\lpcs}{\lpcs+\lpcs} \frac{\delta\mu\sigma}{(\delta +\mu )^2}-\left(\frac{1}{\delta+\mu}+\frac{\delta}{(\delta+\mu)}\frac{1}{2\lpcs}\right)\times\lpcs \kappa\left(\lpcs +\lpcs\right),
    \\ & = \frac{\delta\kappa}{18(\delta+\mu)}\times\left(\delta+\frac{6\mu\sigma}{\kappa(\delta+\mu)}-\sqrt{\delta^2+\frac{3\mu\delta\sigma}{\kappa(\delta+\mu)}}\right).\label{hftpipc2}
\end{align}
Similarly, evaluating $\pi^{i}_{0}$ at $\lods$ and $\ask^\star$, we obtain the equilibrium HFT profit under the decentralized market environment:
\begin{align}\label{hftpiod1}
    \pi_{0}^\text{i}(\lods) & = \frac{\lods}{\lods+\lods}\frac{\delta\mu\sigma}{(\delta+\mu)^{2}}-\frac{\delta}{\delta+\mu}\frac{1}{2\lods}\kappa\lods(\lods+\lods),
    \\ & = \frac{\delta\mu\sigma}{4(\delta+\mu)^{2}}.\label{hftpiod2}
\end{align}
Comparing $\pi_{0}^\text{i}(\lpcs)$ and $\pi_{0}^\text{i}(\lods)$ from equations (\ref{hftpipc2}) and (\ref{hftpiod2}), respectively, we arrive at the following Corollary.

\begin{cor}[HFT Rents from the Speed Race]\label{pi:pcod}
The expected rent earned from the speed race by an HFT is greater in the centralized market than in the decentralized market.
\end{cor}
\begin{proof}
See Appendix.
\end{proof}

On decentralized exchanges, HFTs earn lower rents from speed. 
Figure \ref{fig:hft_rents} presents this result graphically. 
At a first glance, the outcome is surprising since, from Proposition \ref{re:pcod}, HFTs employ fewer infrastructure resources. 
Moreover, equations (\ref{hftpipc2}) and (\ref{hftpiod2}) underscores that HFTs do not extract extra rents from liquidity traders, as the first terms of both equations are equal. 
These effects, however, are dominated by the ``surge-pricing'' feature of decentralized markets, as HFTs start acquiring speed following news arrival--during a ``micro-burst''--where competition for speed intensifies beyond that achieved in a centralized market  (Proposition \ref{si:pcod}). 
Therefore, the on-demand speed race of decentralized markets allows for a transfer of value from HFTs to suppliers of processing speed infrastructure. 

\begin{figure}[H]
\caption{\label{fig:hft_rents} \textbf{HFT Speed Race Rents}}
\begin{minipage}[t]{1\columnwidth}%
\footnotesize
This figure illustrates the equilibrium HFT profits in both centralized and decentralized exchanges, net of technology cost, as a function of the news arrival rate $\delta$. Parameter values: $\mu=2$, $\kappa=0.25$, and $\sigma=1$.
\end{minipage}

\vspace{0.075in}

\begin{centering}
\includegraphics[width=\textwidth]{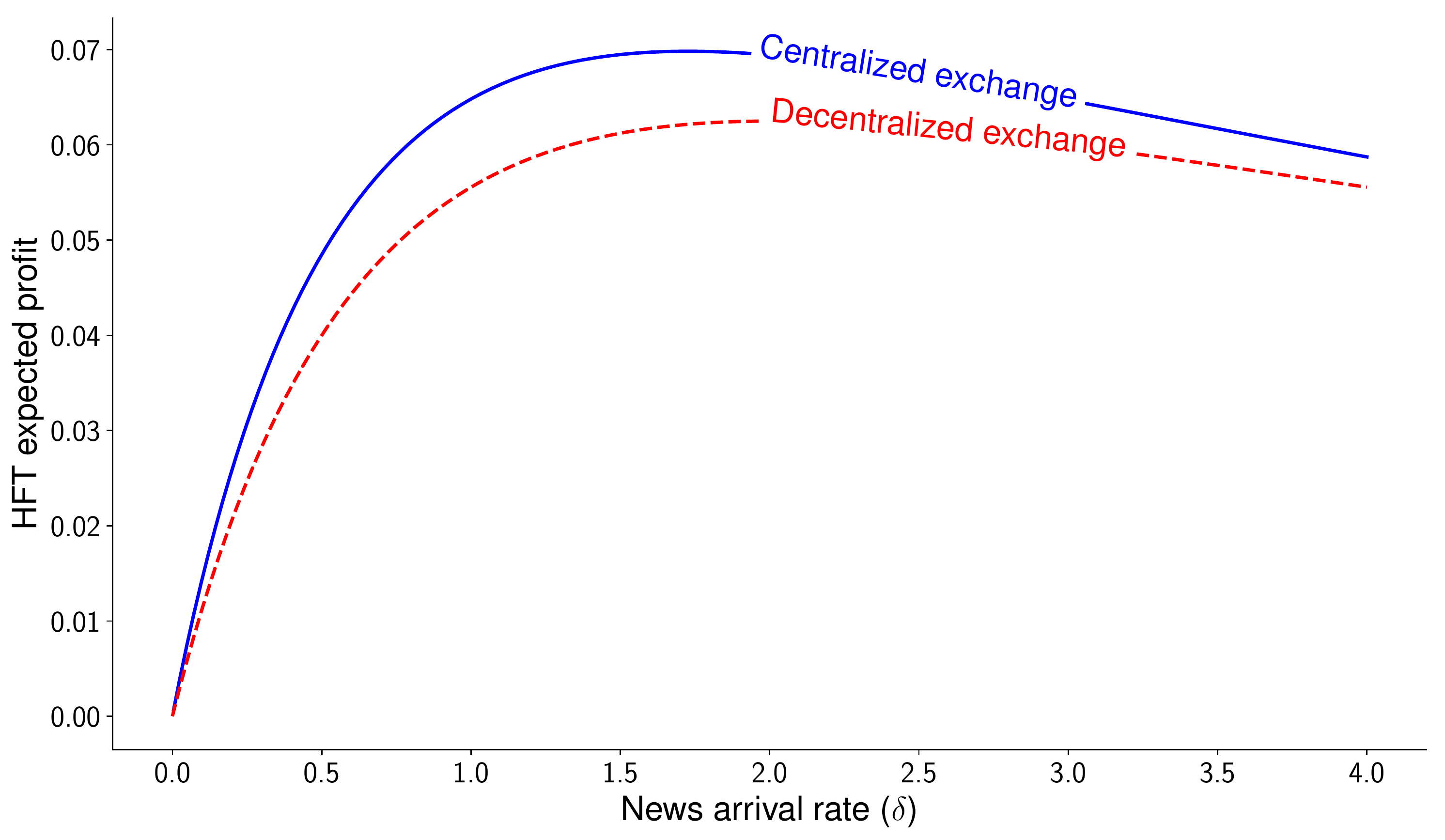}
\par\end{centering}

\end{figure}

\section{Concluding Remarks}\label{sec:conc}

In this paper, we argue that distributed exchanges can mitigate the negative consequences of the high-frequency trading arms race. On decentralized exchanges, HFTs acquire speed in real time, on an as-needed basis, whereas centralized exchanges provision excess capacity to these accommodate ``trading micro-bursts,'' following trading signals. By its nature, a distributed exchange eliminates the negative externality associated with maintaining idle capacity that occurs, for example, with co-location at centralized exchanges. We find that decentralized exchanges encourage short-lived, though intense, HFT races. 

On decentralized exchanges, intense HFT competition for speed during microbursts triggers a surge in the price of computer power. As a result, HFTs earn lower rents from low-latency trading and, at the same time, the overall resource consumption is lower. 

Decentralized exchanges can improve, or at least not harm, traditional measures of market quality. We find relatively quicker price discovery at decentralized exchanges at no cost to liquidity. Quicker price discovery obtains subject to the caveat that cloud-based exchanges exhibit slower stages of the trading process than centralized markets (i.e., from the trader computer to the exchange front end). Therefore, we caution that our model does not insist that decentralized exchanges are the fastest trading mechanism from order submission to settlement, nor do they necessarily yield quicker price discovery than centralized exchanges \textit{overall}. We contend, however, that these stages may see improvements in speed in the near future. This would highlight decentralized exchanges as a viable solution toward reducing the social cost associated with the low-latency arms race, without harming market quality as measured, for example, by liquidity and price discovery.\footnote{In February 2019, Binance (one of the largest cryptoasset exchange) launched a decentralized exchange allowing for one-second settlement times. See: \href{https://techcrunch.com/2019/02/20/binance-dex-decentralized-crypto-exchange/}{Binance releases a first version of its decentralized crypto exchange}.}

\bibliographystyle{jf}
\bibliography{references}


\newpage
\appendix

\numberwithin{equation}{section}
\numberwithin{prop}{section}
\numberwithin{lem}{section}
\numberwithin{defn}{section}
\numberwithin{cor}{section}

\let\normalsize\small

\section{Notation summary\label{sec:Variable-Definitions}}

\noindent
\begin{center}
\begin{tabular}{@{}ll@{}}
\toprule
\cmidrule{1-2}

\multicolumn{2}{c}{\textbf{Variable Subscripts}}\\
\cmidrule{1-2}
Subscript & Definition \\
\cmidrule{1-2}
$M$ & pertaining to HFT market-maker role\\
$B$ & pertaining to HFT ``bandit'' (i.e., sniper) role\\
$\pc$ & ``pre-commitment,'' or centralized market \\
$\od$ & ``on demand,'', or decentralized market \\
\cmidrule{1-2}

\multicolumn{2}{c}{\textbf{Exogenous Parameters}}\\
\cmidrule{1-2}
Parameters & Definition \\
\cmidrule{1-2}
$v$ & asset value at $t=0$, normalized to zero. \\
$\delta$ & Poisson arrival rate of news, i.e., common value innovations. \\
$\mu$ & Poisson arrival rate of liquidity traders (LI). \\
$\kappa$ & elasticity of CPU power supply function. \\
$\sigma$ & absolute size of common value innovations, if there is news. \\
\cmidrule{1-2}

\multicolumn{2}{c}{\textbf{Endogenous Quantities}}\\
\cmidrule{1-2}
Variable & Definition \\
\cmidrule{1-2}
ask & The price at which the market-maker is willing to sell one unit of the asset at $t=1$. \\
$\lambda_i$ & HFT speed, i.e., the Poisson intensity of the HFT market arrival process. \\
$\Lambda$ & Total expenditure with market CPU infrastructure. \\
$p\left(\cdot\right)$ & Market price of CPU infrastructure, $p=\kappa\sum_i \lambda_i$. \\
$\pi^0\left(\cdot\right)$ & HFT profit from trading, net of speed cost. \\

\bottomrule
\end{tabular}
\end{center}

\section{Proofs \label{sec:proofs}}

\noindent \textbf{\large Proposition \ref{si:pcod}}
\begin{proof}
To show that $\lods>\lpcs$, we compute the difference in speed intensity across the two market settings $\lods-\lpcs$,

\begin{equation}
    \lods - \lpcs = \frac{\delta+\frac{3\mu\sigma}{2\kappa(\delta+\mu)}-\sqrt{\delta^2+\frac{3\mu\sigma}{\kappa(\delta+\mu)}}}{6}
\end{equation}
It is enough to show that the numerator is positive, which we obtain by showing the following,
\begin{align}
    \delta+\frac{3\mu\sigma}{2\kappa(\delta+\mu)} & >\sqrt{\delta^2+\frac{3\mu\sigma}{\kappa(\delta+\mu)}}, 
    \\ \iff \left(\delta+\frac{3\mu\sigma}{2\kappa(\delta+\mu)}\right)^2 & >\left(\sqrt{\delta^2+\frac{3\mu\sigma}{\kappa(\delta+\mu)}}\right)^2 \nonumber
    \\ \iff \delta^2+\frac{3\mu\sigma}{\kappa(\delta+\mu)} + \frac{9\mu\sigma}{4\kappa(\delta+\mu)} &> \delta^2+\frac{3\mu\sigma}{\kappa(\delta+\mu)} \nonumber
    \\ \iff \frac{9\mu\sigma}{4\kappa(\delta+\mu)} &>0\nonumber
\end{align}
which concludes the proof. \end{proof}

\noindent \textbf{\large Corollary \ref{pi:pcod}}
\begin{proof}
We compute the difference between the HFT expected profit in the two market settings from equations \eqref{hftpipc2} and \eqref{hftpiod2}, that is
\begin{equation}
    \pi_{0}^\text{i}(\lpcs)-\pi_{0}^\text{i}(\lods)=\frac{\delta}{36\left(\delta+\mu\right)^2} \left[2 (\delta +\mu ) \left(\delta  \kappa -\sqrt{\delta  \kappa  \left(\delta  \kappa
   +\frac{3 \mu  \sigma }{\delta +\mu }\right)}\right)+3 \mu  \sigma\right].
\end{equation}
It is enough to show that 
\begin{equation}
    f\left(\sigma\right)\equiv 2 (\delta +\mu ) \left(\delta  \kappa -\sqrt{\delta  \kappa  \left(\delta  \kappa
   +\frac{3 \mu  \sigma }{\delta +\mu }\right)}\right)+3 \mu  \sigma >0.
\end{equation}
We note that $f\left(0\right)=0$ and further that $f$ increases in $\sigma$ since:
\begin{equation}
    \frac{\partial f}{\partial \sigma}=3 \mu  \left(1-\underbrace{\frac{\delta  \kappa }{\sqrt{\delta  \kappa  \left(\delta  \kappa
   +\frac{3 \mu  \sigma }{\delta +\mu }\right)}}}_{<1}\right)>0,
\end{equation}
which concludes the proof. \end{proof}

\end{document}